\documentclass[aps,physrev,twocolumn,superscriptaddress]{revtex4-2}
\usepackage{mainStyle}

\begin{document}

\title{How much randomness in a quantum process can be explained using memory?}

\author{Derek D. C. Chang}
\email[]{dere0005@e.ntu.edu.sg}
\affiliation{Nanyang Quantum Hub, School of Physical and Mathematical Sciences, Nanyang Technological University, Singapore.}
\affiliation{Centre for Quantum Technologies, Nanyang Technological University, Singapore.}

\author{Graeme D. Berk}
\email[]{graemedean.berk@ntu.edu.sg}
\affiliation{Nanyang Quantum Hub, School of Physical and Mathematical Sciences, Nanyang Technological University, Singapore.}
\affiliation{Centre for Quantum Technologies, Nanyang Technological University, Singapore.}

\author{Mile Gu}
\email[]{gumile@ntu.edu.sg}
\affiliation{Nanyang Quantum Hub, School of Physical and Mathematical Sciences, Nanyang Technological University, Singapore.}
\affiliation{Centre for Quantum Technologies, Nanyang Technological University, Singapore.}
%\homepage[]{Your web page}
%\thanks{}
%\altaffiliation{}

% \date{\today}

\begin{abstract}For a stochastic process describing observations of a dynamical system, complexity science provides systematic methods to decompose the information produced into true irreducible randomness, and that which corresponds to structure superficially disguised as random but can in fact be learned and predicted. Such methods then equip vital tools for prediction, control and inference of the system's internal structure. However, quantum analogs remain underdeveloped due to the inherent complications of invasive measurements and quantum correlations. By harnessing the Choi state representation of process tensors which encode multi-time input-output relations, we formulate a convergent measure of irreducible randomness in stationary quantum stochastic processes. This directly enables a measure of temporal correlations which lower bounds the memory resources required to replicate the process via a recurrent quantum circuit.\end{abstract}

\maketitle

To an observer with limited memory, even highly structured processes can superficially appear random. Distinguishing this finite-memory induced uncertainty from true irreducible randomness is central to prediction, simulation, and control of stochastic processes. In the classical setting, two quantities organize this distinction. The entropy rate quantifies the asymptotic rate of information production~\cite{shannon1948mathematical,cover_elements_2005,crutchfield_regularities_2003}, while the excess entropy quantifies the information shared between the past and future and lower bounds the memory required by the optimal predictive model~\cite{lindgren1988complexity,bialekPredictabilityComplexity2001,crutchfield_regularities_2003,shalizi_computational_2001,crutchfield_times_2009}.
Together, these quantities constrain achievable performance in predictive modeling and learning~\cite{shalizi_computational_2001,bialekPredictabilityComplexity2001}, the thermodynamic manipulation of temporal patterns~\cite{garner_thermodynamics_2017,boyd_correlation-powered_2017}, and facilitate the structural classification of systems between order and chaos~\cite{crutchfield_between_2012}.

We anticipate the utility of these ideas to carry over to the quantum regime, where analogous tools continue to grow in importance. Physical quantum devices inevitably interact with their environments, and therefore operate as open systems, often retaining memory of earlier interactions~\cite{breuer2002theory,breuer2016colloquium,white2020demonstration,morris2022quantifying,karpat2021synchronization}. When memoryless models are adopted, and such effects are neglected, benchmarking and characterization of processes become more challenging, and non-Markovian structure is mistaken for noise~\cite{figueroa2021randomized,figueroa2022towards,
f_kam_detrimental_2025,kam2026spatiotemporal}. Thus, identifying randomness that is merely superficial can enable memory-aware characterization and control, noise-adapted error correction, and non-Markovian noise mitigation~\cite{white2020demonstration,delben2023control,biswas2025noise,wang2025non}. This motivates two basic questions: \emph{how much of the randomness in a quantum process can be explained away, and what memory resources are required to do so?}

Here, we answer these questions by introducing the quantum entropy rate and quantum excess entropy of fully quantum stochastic processes---those whose inputs and outputs are quantum states. 
Cleanly separating irreducible from superficial randomness, the quantum excess entropy arises as the cumulative finite-window overestimate of the irreducible quantum entropy rate, which we show is equal to the quantum mutual information shared between the past and future.
Moreover, we illustrate the physical significance of the quantum excess entropy through a lower bound on the memory resources needed by a recurrent quantum circuit to correctly model the process. 

Our approach addresses gaps in previous attempts to characterize irreducible quantum randomness, in the face of distinct conceptual hurdles due to the interactive nature of quantum systems. An experimenter's choice of control inevitably perturbs the system. Hence, previous studies have either fixed this choice and characterized instead the resulting measured classical process~\cite{crutchfieldIntrinsicQuantumComputation2008, wiesnerNatureComputesInformation2010,suen2017classical,venegas-liMeasurementinducedRandomness2020,venegas2023optimality}, or required optimization over controls~\cite{lindbladNonmarkovianQuantum1979, lindbladQuantumErgodicityChaos1986,connesDynamicalEntropyAlgebras1987,alickiInformationtheoreticalMeaning2002}. 
Meanwhile, an approach agnostic to the choice of control characterizes time series of non-entangled qudits generated by a classical process~\cite{gier_stochastic_2025}. 
Here, we analyze the natural case of underlying quantum dynamics with quantum temporal correlations, via the choice of input states that preserves all information of the process.

\section{Quantum Stochastic Process}
\begin{figure}[ht] \includegraphics[width=0.95\linewidth]{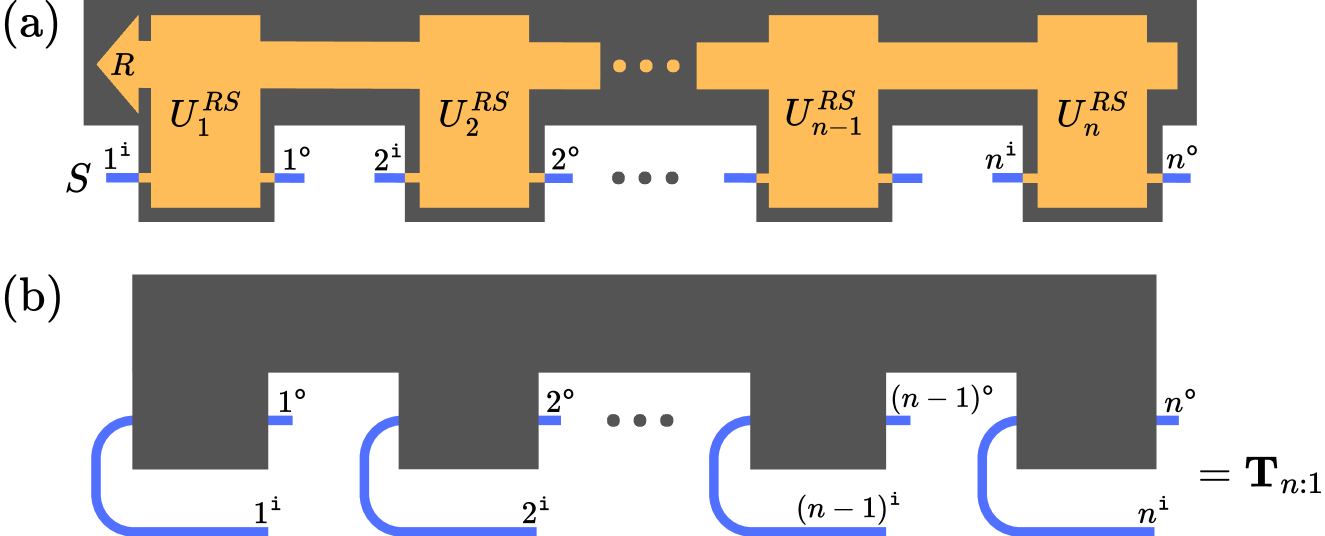} 
    \vspace{-5pt} \caption{\label{fig:ptChoi} 
    (a) A quantum stochastic process encodes all multi-time input-output relations of a system $S$ undergoing open quantum dynamics. The underlying dynamics (orange) may be described by recurrent global unitary evolution of the system and its complement environment $R$. An experimenter may interact arbitrarily with $S$ at each time-step $t$ via input and output arms, for example by state preparation on $t^{\mathtt{i}}$ and measurement on $t^{\mathtt{o}}$. 
    (b) The Choi state $\T_{n:1}$ of a finite window process as the $2n$-body state (blue arms) that results from feeding half a maximally entangled state into $S$ at each step.
    }
\end{figure} 
Our formulation uses the process tensor framework to describe non-Markovian quantum stochastic processes ~\cite{chiribella_theoretical_2009,milz_introduction_2017,
pollock_non-markovian_2018,milz_quantum_2021}. Consider a system $S$ with Hilbert space $\mathcal{H}^{S}$ of fixed dimension $d_S$ that undergoes open quantum dynamics (see Fig.~\ref{fig:ptChoi}a), with which an experimenter can interact repeatedly over multiple time-steps~\footnote{Strictly speaking, each time step is a temporal episode over which the interaction occurs.}. 
All possible multi-time input-output statistics are encoded as a process tensor, which can be uniquely represented by its Choi state, obtained by inputting one arm of a maximally entangled state $\ket{\Phi} \!\equiv\! \sum_{j=0}^{d_S -1} \ket{j}\ket{j} / \sqrt{d_S}$ to $S$ at each time-step $t$ (see Fig.~\ref{fig:ptChoi}b). 
We adopt right-to-left ordering of time indices, where $\T_{n:1}$ denotes a process Choi state over discrete time-steps $1,2,\dots,n$, in alignment with existing literature. 
Denoting the processed arm by $t^{\mathtt{o}}$ and the unprocessed arm by $t^{\mathtt{i}}$, 
$\T_{n:1}$ is then the $2n$-body quantum state over the Hilbert space $\otimes_{t=1}^{n}(\mathcal{H}^{t^{\mathtt{o}}} {\otimes}\mathcal{H}^{t^{\mathtt{i}}})$.  We do not distinguish between the process and its Choi state and consider only normalized Choi states unless otherwise specified.

This process-state equivalence allows us to study temporal properties of the process using information-theoretic tools developed for spatial quantum states. The Choi state $\T_{n:1}$ contains all information about the process, with temporal correlations mapped onto spatial correlations between its subsystems. $\T_{n:1}$ is a state representation of a process that is positive and obeys causality trace conditions: 
\begin{align} 
    \tr_{n^{\mathtt{o}}} \T_{n:1} = \frac{\openone_{n^{\mathtt{i}}}}{d_S} \otimes \T_{n-1:1} \quad\forall n\!\in\! \mathbb{Z}^+, 
\end{align} 
\noindent i.e., there is no operation on output $n^{\mathtt{o}}$ that can influence the past process $\T_{n-1:1}$, reflecting no-signaling backward-in-time. 

Since we are interested in intrinsic structural properties of quantum processes, it is natural to consider those that are stationary, i.e., invariant under time-translation. 
Physically, we take the perspective that an underlying bi-infinite quantum stochastic process $\T$ exists, for which any length-$n$ time window is a finite process $\T_{n:1}$. 
A quantum stochastic process $\T$ is \emph{stationary} if and only if all finite windows are time-translationally invariant:
\begin{align} 
    \T_{n:1} = \T_{n+t:1+t} \quad\forall n\!\in\! \mathbb{Z}^+, t\!\in\! \mathbb{Z}^+ \label{eq:stationarity},
\end{align}
where $\T_{n+t:1+t} \equiv \tr_{t:1}\T_{n+t:1}$ such that accessing a stationary process across any of its finite windows produces equivalent statistics, assuming preceding time-steps were driven with maximally uncertain inputs.
Henceforth, we assume all quantum stochastic processes mentioned are stationary.

\section{Irreducible Randomness}
Classically, the irreducible randomness is quantified by the entropy rate---the average rate of Shannon entropy growth as longer sequences are considered. Given a classical stationary stochastic process governed by random variables $\ldots, X_1, \ldots, X_n, \ldots$, its entropy rate is $\lim_{n\to \infty} \frac{1}{n} H(X_{n:1})$, where $H(X_{n:1})$ is the Shannon entropy of a length-$n$ sequence. It satisfies the key identity
\begin{equation}
    \lim_{n\to \infty} \frac{1}{n} H(X_{n:1}) = \lim_{n\to \infty} H(X_{n}|X_{n-1:1}),\label{eqn:erate}
\end{equation}
\noindent such that it also represents the persisting conditional entropy of the next symbol given the infinite past. 
Fair coin flips have an entropy rate one bit per flip, irreducibly maximal even with unlimited historical data. Meanwhile, a deterministic sequence $\ldots 001001 \ldots$ that emits a $1$ every two $0$s has zero entropy rate, as an observer who has seen a sufficiently long sequence will know the next bit with certainty. See Appendix~\ref{app:classical} for a review of classical entropy rate.

We seek to extend applicability of this fundamental idea into the realm of quantum stochastic processes. 
However, doing so is complicated by the inherently input-output nature of quantum processes. Each round of interaction with $\T$ involves (i) feeding in some input state and (ii) measurement in some basis, both of which can be freely chosen by the experimenter. A prior approach to circumvent this is to fix a specific sequence of experimenter controls, then directly apply classical measures of entropy rate to the resulting sequence of measurement outcomes~\cite{crutchfieldIntrinsicQuantumComputation2008, wiesnerNatureComputesInformation2010,suen2017classical,venegas-liMeasurementinducedRandomness2020,venegas2023optimality}. 
While this approach lets us use classical tools immediately, it does not capture the full picture of quantum entropy rates. Consider a process that applies a Hadamard gate to each input. Since this channel is reversible, one might expect no intrinsic randomness, yet computational-basis inputs and measurements yield a completely random binary sequence with maximal entropy rate. Only by choosing to measure in the complementary basis do we get a classical outcome sequence the expected entropy rate of $0$.

Our approach is to characterize instead the entropy induced on the dilating environment by the dynamics of completely unbiased inputs.
Consider an experimenter implementing uniform probing at each timestep, such that each input $\ket{x_t}$ is drawn from an ensemble of the maximally mixed state.
This avoids bias that masks a process's true behavior---a process that preserves $\ket{0}$ but completely depolarizes $\ket{1}$ would in no way appear random if the experimenter only inputs $\ket{0}$.
We thus ask: how much entropy is induced by such a scheme on an initially pure inaccessible environment $R$ complement to system $S$ (see Fig.~\ref{fig:ptChoi}a)~\footnote{Note that $\mathcal{H}_R$ denotes the environmental degrees of freedom spanning the complement of system $S$ such that $\mathcal{H}_{\text{total}} \equiv \mathcal{H}_S \otimes \mathcal{H}_R$.}? 
Consider first the case where the process is Markovian, such that $\T_{n:1} = \otimes_{t=1}^{n}\T_{t:t}$ where $\T_{t:t}$ is the Choi state of a channel. 
The resulting entropy induced on the pure environment is its von Neumann entropy $S(\T_{t:t})$, which aligns with the map entropy~\cite{zyczkowski2004duality,roga2011entropic,rogaUniversalBoundsHolevo2010,rogaEntropicTradeoff2013,vinjanampathyEntropyBounds2015,kurzykRelatingEntropies2021}---a well-established measure of entropy for quantum channels that is zero if and only if said channel is unitary. Thus, the aforementioned Hadamard process is assigned zero entropy rate.

We observe that this definition naturally extends to any finite window $\T_{n:1}$ of a quantum stochastic process $\T$, where the entropy induced on a pure environment over $n$ time-steps is $S(\T_{n:1})$. We can then define the quantum entropy rate as the average entropy induced in the limit $n \to \infty$.
\begin{definition}\label{def:quantumentropyrate}
    The \emph{quantum entropy rate} of a quantum stochastic process $\T$ is defined as
    \begin{align}
        s_{\T} &\equiv \lim_{n\to\infty} \frac{S(\T_{n:1})}{n}.
        \label{eq:quantumentropyrate}
    \end{align}
\end{definition}
\noindent We note that this quantity was first introduced in~\cite{dowling_operational_2024} as a dynamical entropy in the context of quantum chaos~\cite{lindbladQuantumErgodicityChaos1986,alickiQuantumDynamical2001b}. However, its convergence, an essential property of stationary processes, has not been established. Moreover, a clear notion of irreducible randomness warrants equality with a quantum analog of the conditional form in Eq.~\ref{eqn:erate}.

However, recall that quantum conditional entropy can be negative for spatially-entangled subsystems. One may ask: when considering state representations of temporal processes such as $\T_{n:1}$, can the conditional process entropy $S(\T_{n}\mid \T_{n-1:1}) \equiv S(\T_{n:1}) - S(\T_{n-1:1})$ also be negative~\footnote{We denote $\T_{n} \equiv \T_{n:n} \equiv \tr_{n-1:1}\T_{n:1}$}? We show in Appendix~\ref{app:proofnonnegativecondentropy} that this negativity is in fact forbidden by no-signaling backward-in-time. 
\begin{restatable}{proposition}{nonnegativecondentropy}\label{prop:nonnegativecondentropy}
    \begin{align}
        S(\T_{n}\mid \T_{n-1:1}) \geq 0, \label{eq:nonnegativecondentropy}
    \end{align}
    i.e. $S(\T_{n:1})$ is monotonically non-decreasing in $n$.
\end{restatable}
\noindent Indeed, unphysical negative growth in process entropy would imply that information is removed from (instead of induced in) the environment. In Appendix~\ref{app:proofquantumentropyrateequivalence}, we use the non-negativity of conditional process entropy to prove its convergence to a finite value in the limit $n \to \infty$, which coincides with the quantum entropy rate.
\begin{restatable}{theorem}{quantumentropyrateequivalence}\label{thm:quantumentropyrateequivalence}
    The quantum entropy rate $s_{\T}$ is finite and
    \begin{align}
        s_{\T} = \lim_{n\to\infty} S(\T_{n}\mid \T_{n-1:1}), \label{eq:stationaryquantumentropyrate}
    \end{align} 
    i.e., the process entropy growth rate $S(\T_{n:1}) - S(\T_{n-1:1})$ converges to $s_{\T}$.
\end{restatable}
\noindent This result illustrates $s_{\T}$ as a clear quantum analog of irreducible randomness. Note that non-zero $s_{\T}$ represents the extra entropy induced in the unbounded environment when we evolve the quantum process an extra time-step in the asymptotic limit where $n \to \infty$ time-steps have already occurred. 

Indeed, consider an experimenter who wishes to reduce entropy introduced to the environment by recalling all available information of the process in the past. 
To do so, they note every interaction with the process for time-steps $n{-}1\!:\!1$. Let $x_t$ be the label of the state they input to $\T$, and $y_t$ be the measurement outcome of some POVM. 
Let $z_t \!\equiv\! (x_t,y_t)$, $p_{\overleftarrow{z}}$ be the probability of obtaining outcome sequence $\overleftarrow{z} \equiv z_1,z_2,\ldots z_{n-1}$ and $\T_n^{\overleftarrow{z}}$ the resulting final-step process Choi state conditioned on this sequence. The data processing inequality then implies
\begin{align}
    s_{\T}
    \leq
    S\!\left(\T_n\mid\T_{n-1:1}\right)
    \leq
    \sum_{\overleftarrow{z}}
    p_{\overleftarrow{z}}
    S\!\left(\T_n^{\overleftarrow{z}}\right).
    \label{eq:quantumentropyratebound}
\end{align}
The first inequality follows from monotonic convergence established in Theorem~\ref{thm:quantumentropyrateequivalence}, while the second follows because converting the quantum past into a classical record cannot increase its ability to condition the future. 
Thus the resulting quantum channel, from the perspective of the experimenter, will on average have map entropy that equals or exceeds the quantum entropy rate $s_T$. 

\section{Superficial Randomness}
Having established the quantum entropy rate as a measure of irreducible randomness in a quantum stochastic process, we now identify superficial randomness. Our approach parallels the classical definition of excess entropy (see Appendix~\ref{app:excessentropy}).
Consider $S(\T_{n}\mid \T_{n-1:1})$ which serves as a length-$n$ approximation of the quantum entropy rate $s_{\T}$, for a finite interaction window. 
We hence denote $s_{\T_{n:1}} \!\equiv\! S(\T_{n}\mid \T_{n-1:1})$. 
Since $s_{\T_{n:1}}$ is monotonically non-increasing in $n$ (see Appendix~\ref{app:proofquantumentropyrateequivalence}), it overestimates $s_{\T}$.
This occurs because interacting with only a finite window may leave some temporal correlations unaccounted for, and thus erroneously identified as randomness.
The extent to which this occurs is quantified by the gap $\left( s_{\T_{n:1}} \!- s_{\T} \right)$. 
Accumulating this overestimation over all times yields the quantum excess entropy.
\begin{definition}\label{def:quantumexcessentropy}
    The \emph{quantum excess entropy} of a quantum stochastic process $\T$ is defined as 
    \begin{align}
        \E &\equiv \sum_{n=1}^{\infty} \left( s_{\T_{n:1}} - s_{\T} \right). \label{eq:quantumexcessentropy}
    \end{align}
\end{definition}
\noindent That is, $\E$ represents the total superficial randomness beyond $s_{\T}$ that is eventually explained by uncovering temporal correlations across all lengths.

Since the quantum entropy rate $s_{\T}$ is the asymptotic growth rate of $S(\T_{n:1})$, it is natural to also ask: \emph{how do temporal correlations shape the approach to this limit?}
We prove a direct relationship between the quantum excess entropy and the growth of the process entropy $S(\T_{n:1})$ in  Appendix~\ref{app:quantumexcessentropyinterceptproof}.
\begin{restatable}{proposition}{quantumexcessentropyintercept}\label{prop:quantumexcessentropyintercept}
    \begin{align}
        \E = \lim_{n\to\infty}\left[S(\T_{n:1}) - s_{\T}\, n\right]. \label{eq:quantumexcessentropyintercept}
    \end{align}
\end{restatable}
\noindent When finite, $\E$ represents the y-intercept of the linear asymptote of $S(\T_{n:1})$. 
One may interpret $\E$ as the residual between the actual process entropy and the total process entropy expected from a fully Markovian process with the same growth rate $s_{\T}$. 
In this paper, we assume that $\E$ is finite, and leave the analysis of divergent $\E$ for future work.

Leveraging Prop.~\ref{prop:quantumexcessentropyintercept}, we solidify the interpretation of $\E$ as a measure of temporal correlations in a quantum process.
We prove in Appendix~\ref{app:quantumexcessmutualproof} that $\E$ is also equal to the quantum mutual information (QMI) across the past-future causal split of a quantum stochastic process.
\begin{restatable}{theorem}{quantumexcessmutual}\label{thm:quantumexcessmutual}
    Quantum excess entropy is equal to the QMI across asymptotically long past-future halves,
    \begin{align}
        \E &= \lim_{n\to\infty} I(\T_{F^{(n)}} : \T_{P^{(n)}}), \label{eq:quantumexcessentropymutual}
    \end{align}
    where the length-$n$ past $P^{(n)} \equiv n\!:\!1$ and future $F^{(n)} \equiv 2n \!:\! n{+}1$ are expressed with equal lengths for simplicity.
\end{restatable}
\noindent Thus, $\E$ is established as a measure of temporal correlations of a quantum process---one that arises naturally through the clean separation of superficial from irreducible randomness.

\section{Bounding Memory}
\begin{figure}[tb]
    \includegraphics[width=0.95\linewidth]{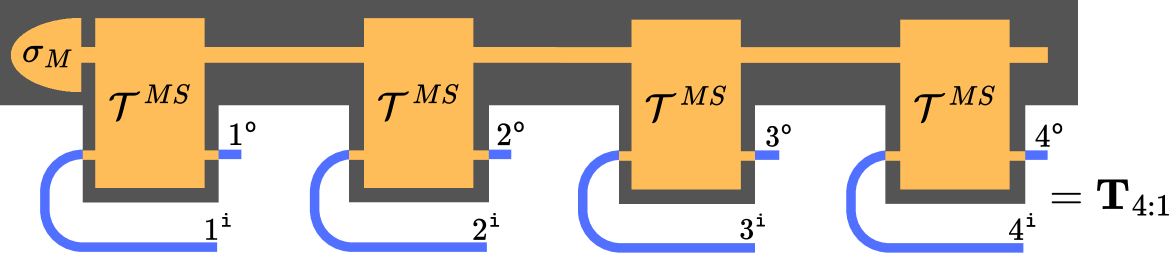} 
    \vspace{-5pt} \caption{\label{fig:ptModel}
    A time-homogeneous recurrent quantum circuit model of the underlying dynamics for a $4$-step window $\T_{4:1}$ of a quantum stochastic process, consisting of a stationary memory state $\sigma_M$ sequentially undergoing fixed system-memory channel $\mathcal{T}^{MS}$. 
    }
\end{figure}

A direct application of quantum excess entropy is its capability to bound the memory resources needed to replicate a quantum process $\T$. In particular, $\T$ is modeled as a quantum system $S$ that undergoes repeated time-homogeneous interactions with a quantum memory system $M$ via some general quantum operation $\mathcal{T}^{MS}$ (see Fig.~\ref {fig:ptModel}). 
The memory system is taken be initialized to stationary state $\sigma_{M} \!\equiv\! \tr_S[\mathcal{T}^{MS}(\sigma_{M} \otimes \openone_S /d_{S})]$ such that it remains invariant under channel $\mathcal{T}^{MS}$ when one remains completely ignorant of the inputs and outputs.
Similar recurrent architectures~\cite{chiribellaQuantumCircuit2008,pollock_non-markovian_2018} appear in quantum reservoir computing~\cite{PhysRevApplied.8.024030}, models of adaptive quantum agents~\cite{elliott_quantum_2022}, and quantum finite-state transducers and generators~\cite{wiesner2008computation}.

A key resource is the amount of information the memory subsystem $M$ must store and propagate from one time-step to the next, which has often been considered a measure of model complexity~\cite{shalizi_computational_2001,crutchfield_synchronization_2010,binder2018practical,thompson_causal_2018}. In classical contexts, where stochastic processes are modeled by finite-state machines, the excess entropy then lower bounds the entropy of their memory states' stationary distribution. In Appendix~\ref{app:excessboundproof}, we establish a quantum analog.

\begin{restatable}{theorem}{excessbound}\label{thm:excessbound}
    Given a quantum stochastic process $\T$ with quantum excess entropy $\E$, and a time-homogeneous recurrent quantum circuit model that reproduces $\T$ with stationary memory state $\sigma_{M}$,
    \begin{align}
        \E &\le 2S(\sigma_{M}).
        \label{eq:excessbound}
    \end{align}
\end{restatable}
\noindent This result reflects a bottleneck---all past information influencing the future must be mediated through the present. 
The bound is saturated when the system-memory dynamics $\mathcal{T}^{MS}$ is unitary and a subsystem of the past purifies the memory state it induces (see Appendix~\ref{app:excessboundproof}). Thm.~\ref{thm:excessbound} establishes that the quantum excess entropy, as a measure of \emph{apparent} memory in a quantum stochastic process, defines a lower bound on the \emph{stored} memory that any model must use to physically realize the process. 
There is an extra factor of $2$ compared to classical excess entropy, which is a consequence of entanglement, allowing a single qubit to mediate up to two bits of mutual information.

\subsection*{Illustrative Example}
\begin{figure*}
    \centering
    \begin{minipage}{0.33\textwidth}
        \centering
        \includegraphics[width=\linewidth]{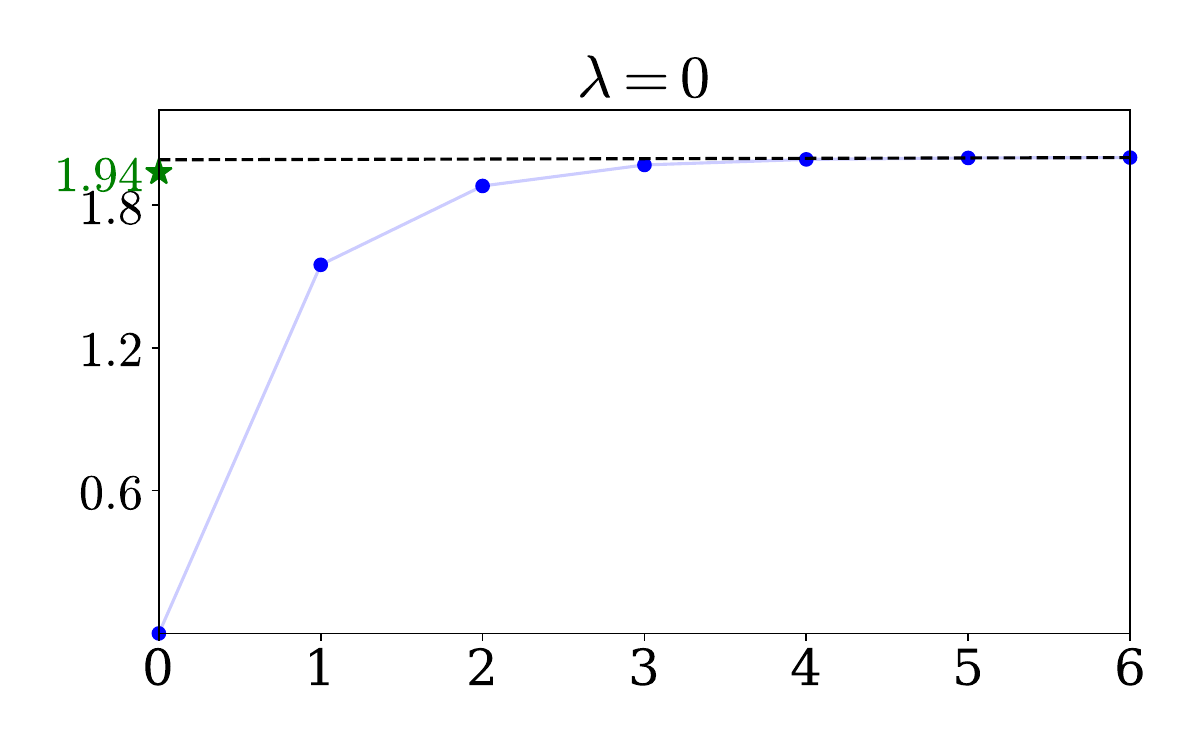}
    \end{minipage}%
    \begin{minipage}{0.33\textwidth}
        \centering
        \includegraphics[width=\linewidth]{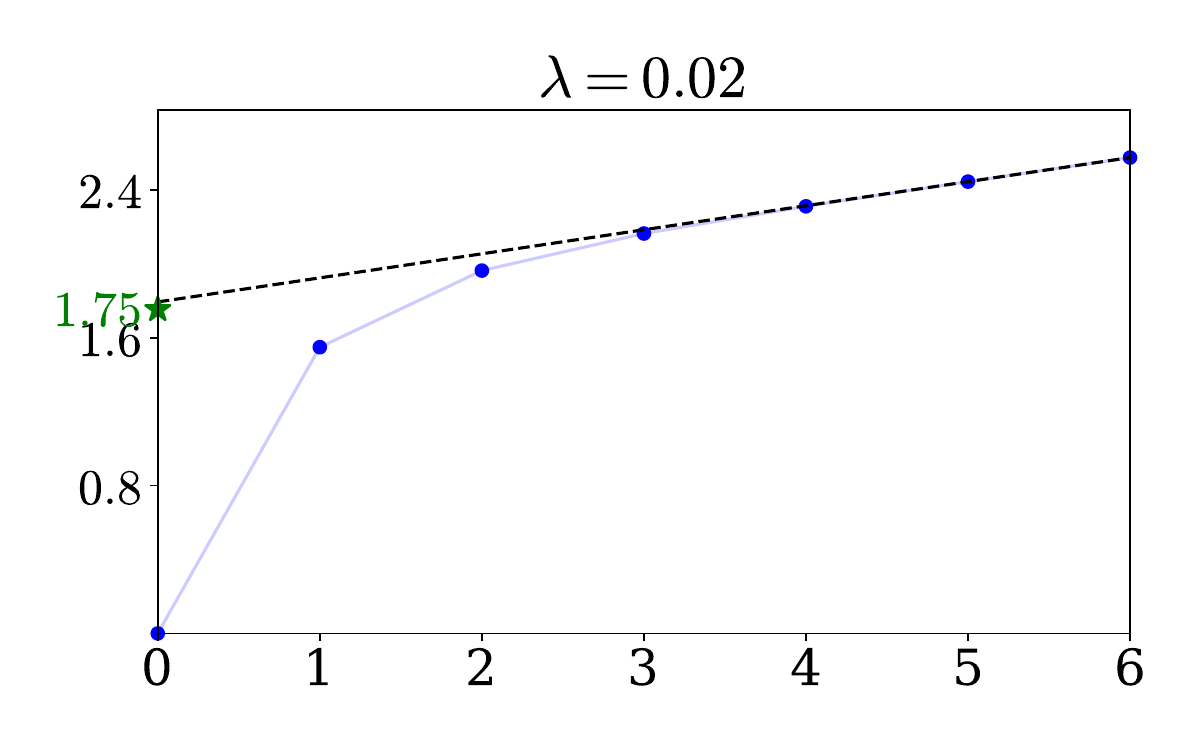}
    \end{minipage}%
    \begin{minipage}{0.33\textwidth}
        \centering
        \includegraphics[width=\linewidth]{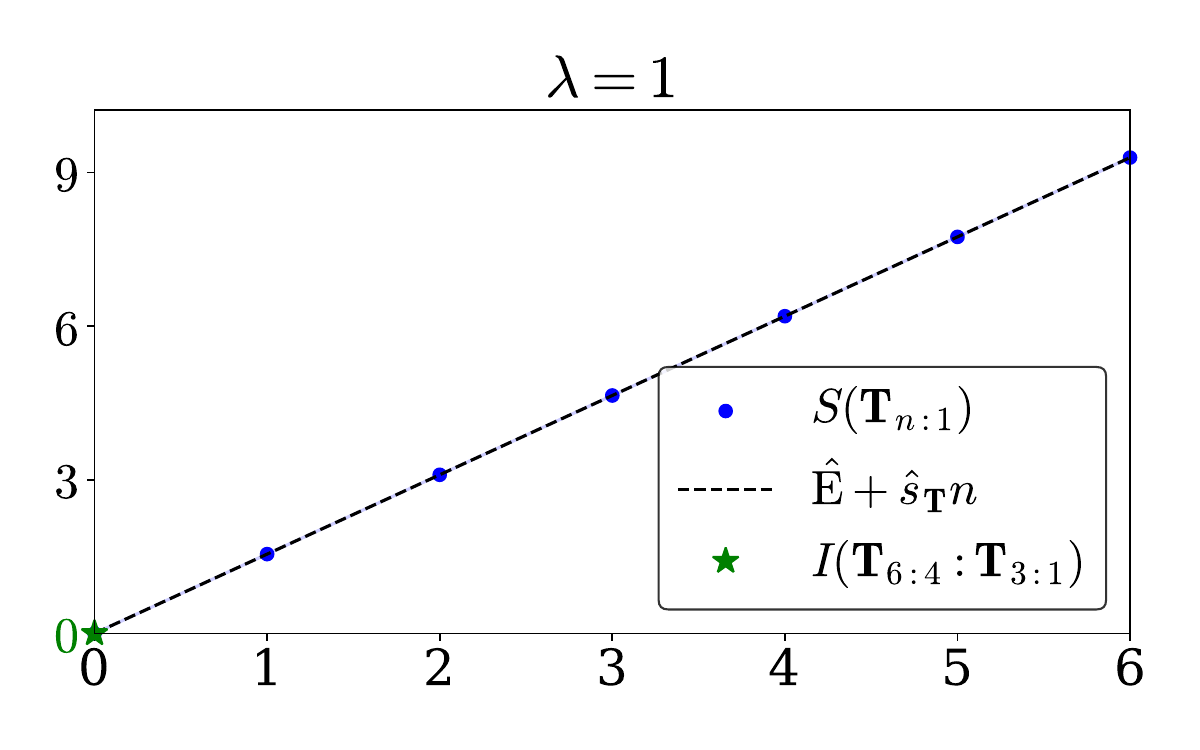}
    \end{minipage}
    \vspace{-10pt}
    \caption{\label{fig:EntropyGrowth}
    $\sqrt{i\text{SWAP}}$ process entropy $S(\T_{n:1})$ (blue points) vs. time-step $n$, for three characteristic regimes: unitary ($\lambda = 0$), partial noise ($\lambda = 0.02$), and full depolarization ($\lambda = 1$).
    Entropy growth rate appears to converge monotonically to the estimated quantum entropy rate $\hat{s}_{\T} = S(\T_6 \mid \T_{5:1})$.
    The quantum excess entropy is estimated as the y-intercept $\hat{\E} = S(\T_{6:1}) - 6\hat{s}_{\T}$ of the linear extrapolation (dashed line), which aligns with $I(\T_{6:4} : \T_{3:1})$ (green star). }
\end{figure*}

To illustrate process entropy growth and quantum excess entropy, we simulate a stationary quantum stochastic process using a recurrent circuit (detailed in Appendix~\ref{app:entropy_growth}).
At each time-step, a system qubit and a memory qubit undergo a unitary $\sqrt{i\text{SWAP}}$, followed by memory depolarization.
Due to $2^{2n}$-dimensional scaling, we compute the process entropy up to $n\!=\!6$ which sufficiently illustrates the qualitative convergence of the process entropy growth rate $S(\T_{n}\mid \T_{n-1:1})$.
We estimate the quantum entropy rate as $\hat{s}_{\T} = S(\T_{6}\mid \T_{5:1})$, and the quantum excess entropy as the linear extrapolation y-intercept $\hat{\E} = S(\T_{6:1}) - 6\hat{s}_{\T}$. 

Fig.~\ref{fig:EntropyGrowth} shows three depolarizing strength ($\lambda$) regimes:
\begin{itemize}
    \item $\lambda \!=\! 0$: The noiseless process exhibits $\hat{s}_{\T} \approx 0$ and saturates the bound in Thm.~\ref{thm:excessbound} with $\hat{\E} \approx 2$ (analytically verified in Appendix~\ref{app:entropy_growth}).
    \item $\lambda \!=\! 0.02$: The intermediate regime demonstrates prototypical monotonic convergence to $\hat{s}_{\T} \!\approx\! 0.13$, with sub-maximal $\hat{\E} \!\approx\! 1.79$.
    \item $\lambda \!=\! 1$: Full depolarization destroys memory-mediated temporal correlations, yielding $\hat{\E} \!=\! 0$ and high $\hat{s}_{\T} \!\approx\! 1.55$.
\end{itemize}
Across all cases, the linear extrapolation suggests approximate convergence to $s_\T$ (Thm.~\ref{thm:quantumentropyrateequivalence}), and alignment of $\hat{\E}$ with the past-future QMI $I(\T_{6:4} : \T_{3:1})$ (Thm.~\ref{thm:quantumexcessmutual}).

\section{Discussion}
Quantum stochastic processes describe diverse phenomena, encompassing essentially any system that repeatedly interacts with an environment or memory system. Examples include quantum reservoir computers~\cite{mujal2021opportunities}, quantum agents~\cite{dunjko2016quantum,elliott2022quantum}, data-re-uploading~\cite{perez2020data}, quantum strategies~\cite{Gutoski_2007}, quantum world models~\cite{lumbreras2026irreducible,jiang2026quantum}, adaptive sensing~\cite{yang2019memory,liu2024fully}, collisions models~\cite{strasberg2019}, sequential entanglement generation~\cite{schon2005sequential} and adaptive quantum error correction~\cite{tanggara2024strategiccodeunifiedspatiotemporal}. Here we provide a universal framework that isolates true irreducible randomness (quantum entropy rate) from structured temporal correlations (quantum excess entropy) with such systems. We show that the quantum excess entropy, as the accumulated superficial randomness beyond the irreducible quantum entropy rate, is equivalent to the past-future quantum mutual information. 
Furthermore, it dictates the fundamental minimum memory resources required by any time-homogeneous recurrent quantum circuit to model the underlying process.
This explicitly links the temporal correlations of a quantum process to the physical resources required for its simulation. 

We expect that our measures of irreducible randomness and temporal correlations will be essential for uncovering the thermodynamic properties of quantum processes.
For example, in work extraction~\cite{Skrzypczyk2014,PerarnauLlobet2015} from quantum sources driven by classical stochastic processes, the entropy rate of the source sets a fundamental upper bound on the extraction rate, while the excess entropy sets the minimum required thermodynamic cost to correlate the agent with the quantum outputs~\cite{huang_engines_2023}. More broadly, the metrics introduced here provide a foundation to establish the fundamental limits on work extraction from fully quantum stochastic processes, through a generalization of a quantum information processing Second Law~\cite{huang_engines_2023,parrondo_thermodynamics_2015,pnas.1204263109,boyd_identifying_2016,boyd_transient_2017,boyd_correlation-powered_2017,PhysRevResearch.2.033334,PhysRevE.105.054131}. 
Recently, non-Markovianity has been identified as a thermodynamic resource that bounds extractable work from general quantum processes~\cite{zambon2025quantum}. It remains to be seen how the quantum excess entropy relates to this bound.

In the spirit of~\cite{crutchfield_regularities_2003}, our measures could similarly pave the way for a structural classification of the complexity of quantum stochastic processes, analogous to the ``spectrum between order and chaos"~\cite{crutchfield_between_2012} found in classical dynamical systems.
A particularly compelling aspect is the study of processes for which process entropy contains a component that grows sublinearly without bound. The resulting infinite excess entropy suggests the presence of unbounded, long-range temporal correlations~\cite{bialekPredictabilityComplexity2001,crutchfield_synchronization_2010}. 
For quantum processes, this necessitates an infinite-dimensional (or continuously expanding) memory.
Quantum models featuring tree-like hierarchical memory structures~\cite{dowling2024capturing}, which obey a power-law decay of these long-range correlations, offer an initial glimpse of highly complex behavior in quantum processes with possibly diverging quantum excess entropy. 

Parallel to complexity, identifying quantum chaos, thermalization, or information scrambling in open quantum systems is also of interest. 
For instance, as a measure of temporal correlations, the quantum excess entropy may be used to characterize the complexity of the influence matrix~\cite{lerose2021influence,cerezo2025spatio} which encodes the local dynamics of a subsystem embedded in a larger many-body system, reminiscent of the process tensor. There, the characteristic scaling behavior of the central existing measure, temporal entanglement (TE)~\cite{lerose2021scaling,giudice2022temporal,foligno2023temporal,taranto_higher-order_2025,vilkoviskiy2025temporal}, has been observed to identify quantum chaotic dynamics~\cite{lerose2021scaling,giudice2022temporal,foligno2023temporal}. 
Recent work shows a close relationship between the scaling of TE and temporal correlations~\cite{vilkoviskiy2025temporal}, but the role of the quantum excess entropy remains to be explored.

\begin{acknowledgments} 
This work is supported by the National Research Foundation of Singapore through the NRF Investigatorship Program (Award No. NRF-NRFI09-0010), the National Quantum Office, hosted in A*STAR, under its Centre for Quantum Technologies Funding Initiative (S24Q2d0009), the RIE 2025 AQAS projects S25Q9D001 and S25Q9D002, the Singapore Ministry of Education Tier 1 Grant RT4/23 and RG91/25 and the RIE25 Japan-Singapore Joint Call on Quantum (Project ID H25-MRO3490).
\end{acknowledgments}

\bibliography{ref.bib}

\newpage
\onecolumngrid
\appendix
\section{Classical Randomness and Memory\label{app:classical}}

We briefly review key information-theoretic quantities that characterize the structural properties of classical stationary stochastic process governed by random variables $\ldots, X_1, \ldots, X_n, \ldots$.
Each random variables $X_t$, has realization $x_t$ in some alphabet $\mathcal{X}$. The stochastic process is characterized by the joint probability distributions of all finite-length sequences $P(X_{n:1})$ for all $n \in \mathbb{Z}^+$, where time indices are understood to have right-to-left ordering, i.e., $X_{n:1}$ represents the process over times $1,2,\dots,n$.
A stochastic process is stationary if its joint distributions are invariant under time translation: $P(X_{n:1}) = P(X_{(n+t):(1+t)})$ for all $t \in \mathbb{Z}^+$. 

The quantities and tools from classical information-theory for stochastic processes will be valuable in the exploration of quantum information-theoretic properties of quantum processes. Note that classical input-output processes---which are a more direct structural analog to quantum processes---have also been studied~\cite{barnett_computational_2015}. However, as we detail in the main text, the Choi formalism allows us to represent input-dependence in a form that is more conceptually similar to output-only processes as presented in this section.
The central quantity is the entropy rate, which is a fundamental measure of irreducible randomness in a classical stochastic process.

\subsection{Entropy Rate\label{app:entropyrate}}
The average information or uncertainty associated with a classical system or a random variable is measured by the Shannon entropy~\cite{shannon1948mathematical,cover_elements_2005} of its probability distribution. This static distribution entropy was generalized to a dynamical entropy associated with evolving dynamical systems by Kolmogorov~\cite{kolmogorov1959entropy, kolmogorov1985new} and Sinai~\cite{sinai1959notion}. The dynamical entropy is also called the Kolmogorov-Sinai entropy (KS entropy). Its equivalent quantity in the context of stochastic processes is the entropy rate (or entropy density)~\cite{cover_elements_2005}. Given a discrete stochastic process as a temporal sequence of (possibly correlated) random variables, its entropy rate is given by 
\begin{align}
    h_X \equiv \lim_{n\to \infty} \frac{1}{n} H(X_{n:1}),
    \label{eq:entropy_rate}
\end{align}
where $H(X_{n:1})$ is the joint entropy of a length-$n$ sequence of random variables.
The entropy rate has an interpretation as the rate of information production as longer sequences are considered.
This quantity has important fundamental and practical significance. In the context of dynamical systems, the similar KS entropy has been shown to be a measure of chaos, as it is under some assumptions equal to the sum of positive Lyapunov exponents~\cite{pesin1977characteristic, ott2002chaos}, which measure the rate of separation of nearby trajectories. It was also formalized as the algorithmic complexity of almost every resulting symbolic sequence of a dynamical system~\cite{brudno1983entropy}. The best known operational interpretation of the entropy rate is given by Shannon~\cite{shannon1948mathematical,cover_elements_2005} as the best achievable average compression rate of an ergodic information source. Similarly, it also quantifies a fundamental limit of prediction~\cite{crutchfield_regularities_2003,shalizi_computational_2001}, as the best asymptotically achievable average uncertainty in predicting the next outcome of a process. That is, for stationary stochastic processes, the entropy rate in Eq.~\eqref{eq:entropy_rate} always converges to a finite limit~\cite{cover_elements_2005}, and can be equivalently expressed with the conditional entropy as 
\begin{align}
    h_X &= \lim_{n\to \infty} H(X_{n}|X_{n-1:1}).
\end{align}
This expression has the interpretation of the average uncertainty about the next outcome, given knowledge of the entire past. This characterizes the persistent unpredictability of the stationary stochastic process as it is the information about the process that cannot be obtained from a history of any length. 
A simple example of a process with finite irreducible randomness is one produced by flipping a fair coin. This i.i.d. process has entropy rate $h_X = H(X_n) = 1$ bit/outcome, which is maximal for a binary alphabet. In comparison, a process with correlations between outcomes would have lower entropy rate as longer sequences reduce uncertainty by uncovering hidden correlations. Importantly, we will see that precisely accounting for this excess randomness beyond the entropy rate gives rise to a measure of memory in stochastic processes, called the excess entropy.

\subsection{Excess Entropy\label{app:excessentropy}}
Consider a finite length-$n$ subsequence of a stationary stochastic process. The finite-$n$ approximation~\cite{crutchfield_regularities_2003,feldman_discovering_2022} of the \emph{true} entropy rate $h_X$ is
\begin{align}
    h_{X_{n:1}} \equiv H(X_{n}|X_{n-1:1}).
\end{align}
This approximation overestimates the true entropy rate by an amount $(h_{X_{n:1}} - h_{X})$. This difference is exactly the extra randomness that could be reduced by considering longer sequences which would reveal more correlations with the future after time $n$. The total amount of temporal correlations can be quantified by aggregating the overestimation of randomness over all time steps.
This total overestimation is the excess entropy~\cite{crutchfield1983symbolic,szepfalusy1986entropy,lindgren1988complexity,crutchfield_regularities_2003}
\begin{align}
    E_X \equiv \sum_{n=1}^{\infty} (h_{X_{n:1}} - h_{X}),
\end{align}
which quantifies the amount of information that must be gained from the process before the per-step randomness remaining is minimal (entropy rate). 
Furthermore, the excess entropy can be shown to be the mutual information between the semi-infinite causal halves (past and future) of a stationary stochastic process~\cite{crutchfield_regularities_2003,feldman_discovering_2022}:
\begin{align}
    E_X = \lim_{n\to\infty} I(X_{2n:n+1} : X_{n:1}).
\end{align}
Thus the excess entropy is a sensible quantifier of the apparent memory in a process, given by the information the observable past shares with the observable future. 

Like the entropy rate, the excess entropy has operational significance in predictive modeling. In particular, in the field of computational mechanics, excess entropy has been shown to be a lower bound on statistical complexity~\cite{shalizi_computational_2001,crutchfield_synchronization_2010}. Statistical complexity is the least amount of information about the past that must be stored as memory states by any model that seeks to produce statistically faithful predictions of the future of a process. The $\epsilon$-machine is the provably optimal predictive model of a stochastic process that achieves this minimum memory requirement~\cite{shalizi_computational_2001}.
In general, the stored memory (statistical complexity) is greater than the apparent memory (excess entropy), as ``typical processes encrypt their state information within their observed behavior''~\cite{crutchfield_times_2009}. 

Our proposals for quantum processes will be heavily grounded in the ideas described in this section. More details on classical stochastic processes and their information-theoretic analysis can be found in~\cite{shaw_dripping_1984,grassberger_toward_1986,bialekPredictabilityComplexity2001,crutchfield_regularities_2003,cover_elements_2005,feldman_discovering_2022}.
Translating these information-theoretic tools to quantum stochastic processes is a key objective which enables the corresponding analysis of the structure of those processes.

\section{Proofs}

\subsection{Proof of non-negativity of conditional process entropy (Prop.~\ref{prop:nonnegativecondentropy})}\label{app:proofnonnegativecondentropy}

We show that despite the possible negativity of quantum conditional entropy for general quantum states, the following quantum conditional entropy is always non-negative for a temporal quantum state constructed as the Choi state of any process tensor $\T_{n:1}$.

\nonnegativecondentropy*
\begin{proof}
    This property is a direct consequence of the causality constraints of process tensors given by $\tr_{n^{\mathtt{o}}} \T_{n:1} = \frac{\openone_{n^{\mathtt{i}}}}{d_S} \otimes \T_{n-1:1}$. 
    By additivity of quantum entropy, we have
    \begin{align}
        S(\tr_{n^{\mathtt{o}}} \T_{n:1}) = \log(d_S) + S(\T_{n-1:1}).
    \end{align}
    The Araki-Lieb triangle inequality~\cite{araki1970entropy} for $S(\T_{n:1})$ is
    \begin{align}
        S(\T_{n:1}) &\geq \left|S(\tr_{n^{\mathtt{o}}} \T_{n:1}) - S(n^{\mathtt{o}})\right| \\
        &\geq \log(d_S) + S(\T_{n-1:1}) - S(n^{\mathtt{o}})
    \end{align}
    Since $S(n^{\mathtt{o}}) \leq \log(d_S)$, we have
    \begin{align}
        S(\T_{n}\mid \T_{n-1:1}) \equiv S(\T_{n:1}) - S(\T_{n-1:1}) \geq 0.
    \end{align}
\end{proof}

\subsection{Proof of quantum entropy rate equivalence (Thm.~\ref{thm:quantumentropyrateequivalence})\label{app:proofquantumentropyrateequivalence}}

We seek a form of quantum entropy rate as an asymptotic entropy of the next step of a process given the past process and show that it is equivalent to Def.~\ref{def:quantumentropyrate} because it indeed converges for stationary quantum processes. 

\quantumentropyrateequivalence*

\begin{proof}
    We first show that the conditional process entropy always converges to a finite value for stationary quantum processes, which allows us to establish the quantum entropy rate as the asymptotic conditional process entropy.
    Stationary quantum processes satisfy the property that $S(\T_{n}\mid \T_{n-1:1})$ is non-increasing in $n$:
    \begin{align}
        S(\T_{n}\mid \T_{n-1:1}) &= S(\T_{n+1}\mid \T_{n:2}) \notag \\
        &\geq S(\T_{n+1}\mid \T_{n:1}), 
    \end{align}
    where the equality holds due to stationarity, and the inequality follows due to strong subadditivity. 
    In other words, the stationary process entropy is discrete concave in $n$: $S(\T_{n:1}) \geq [S(\T_{n+1:1}) + S(\T_{n-1:1})] / 2$.
    Since $S(\T_{n}\mid \T_{n-1:1})$ is non-increasing and also lower bounded by $0$ (Prop.~\ref{prop:nonnegativecondentropy}), by the monotone convergence theorem, its limit in eq.~\eqref{eq:stationaryquantumentropyrate} exists. 
    We remark that the same conclusion can be reached by noting the looser lower bound $S(\T_{n}\mid \T_{n-1:1}) \ge -2\log(d_S)$, which is sufficient for convergence.

    Next, note that $S(\T_{n:1})$ can be decomposed as a telescoping sum of conditional process entropies,
    \begin{align}
        S(\T_{n:1}) &= \sum_{t=1}^{n} \left[S(\T_{t:1}) - S(\T_{t-1:1})\right], \label{eq:quantumconditionalentropychain}
    \end{align}
    where $S(\T_{t-1:1})|_{t=1} = 0$ as the process is defined to start from time $t=1$.
    We seek to show that the average entropy growth rate equals the asymptotic entropy growth rate for stationary quantum processes:
    \begin{align}
        \label{eq:quantumentropyrateequivproof}
        \lim_{n\to\infty}\frac{\sum_{t=1}^{n}\left[S(\T_{t:1}) - S(\T_{t-1:1})\right]}{n}
        = \lim_{n\to\infty} \left[S(\T_{n:1}) - S(\T_{n-1:1})\right].
    \end{align}
    Since we have shown that the limit on the right exists, the Ces\`{a}ro mean~\footnote{\label{foot:Cesaro}Ces\`{a}ro mean: Given $b_N = \sum_{n=1}^{N} a_n /N$, if $a_{N} \to a$ for finite $a$, then $b_N \to a$ too. It roughly states that the average of a sequence tends to the limit of a sequence, if the limit exists.} immediately implies the equality.
\end{proof}

\subsection{Proof of quantum excess entropy as sublinear component (Prop.~\ref{prop:quantumexcessentropyintercept})}\label{app:quantumexcessentropyinterceptproof}

We show a useful expression of quantum excess entropy that is required in the proof of Thm.~\ref{thm:quantumexcessmutual}.

\quantumexcessentropyintercept*
\begin{proof}
    \begin{align}
        \mathrm{E} &\equiv \lim_{n\to\infty}\sum_{t=1}^{n} \left( s_{\T_{t:1}} - s_{\T} \right) \\
        &= \lim_{n\to\infty}\sum_{t=1}^{n} \left[S(\T_{t:1}) - S(\T_{t-1:1}) - s_{\T}\right] \\
        &= \lim_{n\to\infty}\left[S(\T_{n:1}) - s_{\T}\, n\right],
    \end{align}
    where the second equality follows by definition, and we use the telescoping sum in the last equality with $S(\T_{t-1:1})|_{t=1} \equiv 0$.
\end{proof}

\subsection{Proof of quantum excess entropy as past-future QMI (Thm.~\ref{thm:quantumexcessmutual})}\label{app:quantumexcessmutualproof}

We strengthen the interpretation of quantum excess entropy as characterizing apparent memory in a quantum process, by showing that, when finite, it is also the quantum mutual information across the past-future causal split of a quantum stochastic process.

\quantumexcessmutual*
\begin{proof}
    Recall from Prop.~\ref{prop:quantumexcessentropyintercept} that the quantum excess entropy is 
    \begin{align}
        \mathrm{E} &= \lim_{n\to\infty}\left[S(\T_{n:1}) - s_{\T}\, n\right] \equiv \lim_{n\to\infty} \mathrm{E}_{n},
    \end{align}
    where $\mathrm{E}_{n}$ is the sublinear component of $S(\T_{n:1})$.
    Note that since the limit converges to finite $\mathrm{E}$, we also have 
    \begin{align}
        \mathrm{E} &= \lim_{n\to\infty}\left[S(\T_{2n:1}) - s_{\T}\, (2n)\right] \equiv \lim_{n\to\infty} \mathrm{E}_{2n}.
    \end{align}
    By definition of QMI, we have
    \begin{align}
        \lim_{n\to\infty} I(\T_{2n:n+1} : \T_{n:1}) &= \lim_{n\to\infty} \left[ S(\T_{2n:n+1}) + S(\T_{n:1}) - S(\T_{2n:1}) \right] \\
        &= \lim_{n\to\infty} \left[ 2S(\T_{n:1}) - S(\T_{2n:1}) \right] \\
        &= \lim_{n\to\infty} \Bigl\{ 2[S(\T_{n:1}) - s_{\T}\, n] - [S(\T_{2n:1}) -  s_{\T}\, (2n)] \Bigr\} \\
        &= \lim_{n\to\infty} (2\mathrm{E}_{n} - \mathrm{E}_{2n}) \\
        &= 2\lim_{n\to\infty}\mathrm{E}_{n} - \lim_{n\to\infty}\mathrm{E}_{2n} \\
        &= \mathrm{E},
    \end{align}
    where the second equality follows from stationarity of the process, and the penultimate equality holds since both limits exist and are finite. Thus, quantum excess entropy is equal to the past-future QMI if their limits are finite.
\end{proof}

\subsection{Proof of quantum excess entropy bound (Thm.~\ref{thm:excessbound})}\label{app:excessboundproof}

Suppose the underlying dynamics of a finite window of a quantum stochastic process $\T$ with finite quantum excess entropy $\mathrm{E}$, is modeled using a (non-unique) recurrent quantum circuit~\cite{chiribella_theoretical_2009, chiribellaQuantumCircuit2008} with a memory subsystem $M$ undergoing sequential fixed time-homogeneous system-memory CPTP dynamics $\mathcal{T}^{MS}$ at every time step.
The memory may be taken to be initialized in a stationary state $\sigma_{M} \!=\! \tr_S[\mathcal{T}^{MS}(\sigma_{M} \otimes \openone_S /d_{S})]$ that remains invariant under the channel $\mathcal{T}^{MS}$ when the system is fed half of the maximally entangled state (or equivalently, when an experimenter's inputs are maximally uncertain).
We illustrate in Fig.~\ref{fig:ptModel} such a recurrent circuit model for an $n\!=\!4$ window. 
Since all past information influencing the future must be stored and mediated through the model's memory state $\sigma_{M}$, we show the following bound of the quantum excess entropy $\mathrm{E}$.

\excessbound*

\begin{proof}
    First consider a finite window $\T_{2n:1} \equiv \T_{F^{(n)}, P^{(n)}}$ with length-$n$ past $P^{(n)} \equiv n{:}1$ and future $F^{(n)} \equiv 2n\!:\!n{+}1$. 
    Let $M_n$ denote the induced memory subsystem immediately after the $n$-th step of the time-homogeneous recurrent circuit model,
    given that the memory is initialized in $\sigma_M$ and half a maximally entangled state is input at every step.
    Note that the overall recurrent CPTP dynamics that takes the stationary memory state $\sigma_M$ after the $n$-th step to the future process $\T_{F^{(n)}}$ may be expressed as an effective quantum channel $\mathcal{N}^{(n)}_{M_n \to F^{(n)}}$ acting locally on $M_n$.
    The past-future QMI of $\T_{F^{(n)}, P^{(n)}}$ then obeys
    \begin{align}
        I(F^{(n)} : P^{(n)}) \le I(M_n : P^{(n)}) \le 2S(\sigma_{M_n}),\label{eq:qmichain}
    \end{align}
    where the first inequality is a due to QDPI~\cite{wilde2013quantum} which requires that correlations between $P^{(n)}$ and $M_n$ cannot increase under local processing on $M_n$ to obtain $F^{(n)}$. 
    The last inequality follows from the definition of QMI and the Araki--Lieb triangle inequality~\cite{araki1970entropy},
    \begin{align}
    I(M_n : P^{(n)}) &= S(M_n) - S(M_n | P^{(n)}) \\
    &\leq S(M_n) - [-S(M_n)] \\
    &= 2S(\sigma_{M_n}).
    \end{align}
    Since the memory state is stationary $S(\sigma_{M_n})=S(\sigma_M)$, taking the limit $n\to \infty$ gives $\E = \lim_{n\to\infty}I(F^{(n)} : P^{(n)}) \le 2S(\sigma_M)$.
\end{proof}
We see that the first inequality in Eq.~\ref{eq:qmichain} saturates in the long time limit $n\to \infty$ when the induced memory $M_n$ carries the same amount of correlations with the asymptotically long past $P^{(n)}$ as the future $F^{(n)}$ does. This is achieved if the effective channel $\mathcal{N}^{(n)}_{M_n \to F^{(n)}}$ is reversible in the long time limit $n\to \infty$, where the system-memory dynamics $\mathcal{T}^{MS}$ for each step is unitary. 
The last inequality is saturated if and only if the memory $M_n$ induced by the past $P^{(n)}$ is purified by a subsystem of the past, i.e., the joint state on $M_n, P^{(n)}$ is of the form  $\ket{\psi}\bra{\psi}_{M_n, P^{(n)}_A} \otimes \rho_{P^{(n)}_B}$ where $\mathcal{H}^{P^{(n)}} = \mathcal{H}^{P^{(n)}_A} \otimes \mathcal{H}^{P^{(n)}_B}$~\cite{Zhang2011on,xi2012necessary}, in the long time limit $n\to \infty$.

\section{Details of illustrative example\label{app:entropy_growth}}

\begin{figure}[h]
\includegraphics[width=0.55\linewidth]{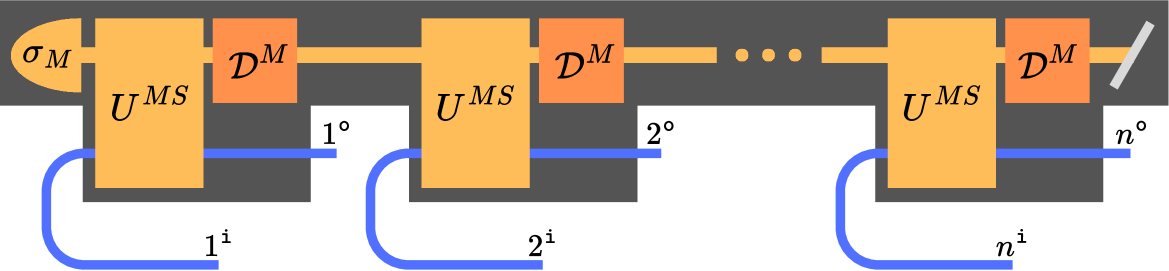}
\caption{\label{fig:EntropyGrowthModel}
    Recurrent quantum circuit model of the depolarized $\sqrt{i\text{SWAP}}$ process.}
\end{figure}

To demonstrate the relationship between the quantum excess entropy and the growth in process entropy, we consider a representative example of a quantum stochastic process with tunable noise. Here we detail the setting of our simulations. 
We compute the Choi state of a process generated by a recurrent quantum circuit with a single qubit memory system $M$, initialized in the maximally mixed stationary state $\sigma_{M} \equiv \frac{\openone}{2}$, interacting sequentially with one half of a Bell state $\ket{\Phi}_{SS'} \equiv \frac{1}{\sqrt{2}}(\ket{0}\ket{0} + \ket{1}\ket{1})$, for $n = 6$ steps. 
This is illustrated in Fig.~\ref{fig:EntropyGrowthModel}. 
The system-memory interaction consists of a unitary $\sqrt{i\text{SWAP}}$, which allows for the transfer of quantum information between the memory and the system to create temporal correlations. This is followed by a depolarizing channel on the memory that introduces noise. The $\sqrt{i\text{SWAP}}$ unitary is given by
\begin{align}
    U^{MS} &\equiv \frac{1}{\sqrt{2}}(\openone + i\text{SWAP})
    ,
\end{align}
where the SWAP operator is defined through its action $\text{SWAP} \ket{\psi}^M \ket{\phi}^S = \ket{\phi}^M \ket{\psi}^S$ for any states $\ket{\psi}$ and $\ket{\phi}$. The depolarizing channel on the memory is defined through its action on an arbitrary state $\rho$ as
\begin{align}
    \mathcal{D}^{M}_\lambda(\rho) = (1-\lambda) \rho + \lambda \frac{\openone}{2},
\end{align}
where we vary the depolarizing parameter $\lambda$ to control the noise strength. The process is stationary by construction, with the memory initialized in the maximally mixed state and the fixed system-memory interaction time-homogeneous. 
We then plot the growth of process entropy $S(\T_{n:1})$ with time for three different depolarizing parameters $\lambda$ for the $\sqrt{i\text{SWAP}}$ process, exemplifying three key classes of qualitative behavior that can be exhibited by a stationary quantum process: $\lambda = 0$ for a noiseless unitary process, $0<\!\lambda\!<1$ for partial noise, and $\lambda = 1$ for complete depolarization. 

\subsection*{Analytical calculation for noiseless case}
The quantum excess entropy for the noiseless $\lambda\!=\!0$ case can be analytically obtained as follows. 
Consider the purification of the initial memory state as $\ket{\Phi}_{M'M} \equiv \frac{1}{\sqrt{2}}(\ket{0}\ket{0} + \ket{1}\ket{1})$ with a reference system $M'$. Since the effective channel on the memory for a single step produces
\begin{align}
    \mathcal{E}(\sigma_M) &= \tr_S \left[ U^{MS} \left( \sigma_M \otimes \frac{\openone_S}{2} \right) \left(U^{MS}\right)^\dagger \right] \\
    &= \frac{1}{2} \sigma_M + \frac{1}{4}\openone_M,
\end{align}
then after $n$ steps, the purification becomes
\begin{align}
    (\openone \otimes \mathcal{E})^{n}\ket{\Phi}\bra{\Phi}_{M'M} &= \left(\frac{1}{2}\right)^{n} \ket{\Phi}\bra{\Phi}_{M'M} + \left(1-\left(\frac{1}{2}\right)^{n}\right) \frac{\openone_{M'M}}{4}
\end{align}
where the first term vanishes as $n\to\infty$. 
Note that for any process modeled by noiseless unitary interaction, the quantum entropy rate is zero, $s_\T = 0$. Thus, due to unitarity of the process, the quantum excess entropy can be expressed as the asymptotic entanglement entropy between $M'M$ and the temporal degrees of freedom, as $\E = \lim_{n\to\infty} S(\T_{n:1}) = S\left(\frac{\openone_{M'M}}{4}\right) = 2$. Hence the quantum excess entropy in this case saturates the bound in Thm.~\ref{thm:excessbound}.

% \twocolumngrid

\end{document}